\documentclass[11pt]{article}

\usepackage[T1]{fontenc}
\usepackage[utf8]{inputenc} 
\usepackage{lmodern}        

\usepackage[margin=1in]{geometry}
\usepackage{microtype}

\usepackage{amsmath,amssymb,amsfonts,amsthm,mathtools}
\usepackage{bm}

\theoremstyle{plain}
\newtheorem{theorem}{Theorem}

\theoremstyle{definition}
\newtheorem{definition}{Definition}
\theoremstyle{remark}

\usepackage{graphicx}
\usepackage{caption}
\usepackage{subcaption}

\usepackage{booktabs}   
\usepackage{array}
\usepackage{tabularx}   
\usepackage{longtable}  
\usepackage{multirow}
\usepackage{colortbl}   
\usepackage[table]{xcolor} 

\usepackage{float} 
\usepackage{placeins} 

\usepackage{enumitem}

\usepackage{tikz}
\usetikzlibrary{arrows.meta,positioning,fit,calc,decorations.pathreplacing}

\usepackage[numbers,sort&compress]{natbib}

\usepackage{hyperref}
\usepackage[nameinlink,capitalise,noabbrev]{cleveref}

\hypersetup{
	colorlinks=true,
	linkcolor=blue,
	citecolor=blue,
	urlcolor=blue,
	pdftitle={},
	pdfauthor={},
	pdfkeywords={}
}

\usepackage{url}

\title{Your Title Here}
\author{Your Name(s) Here}
\date{} 

\makeatletter
\@ifundefined{bbbk}{}{%
	
}
\makeatother

\title{AI Social Responsibility as Reachability:
	Execution-Level Semantics for the Social Responsibility Stack}

\author{Otman Adam Basir\\
	Department of Electrical and Computer Engineering\\
	University of Waterloo, Waterloo\\
	Canada\\
	email:obasir@uwaterloo.ca
	}

\makeatletter
\@ifundefined{bbbk}{}{%
	
}
\makeatother

\title{AI Social Responsibility as Reachability:
	Execution-Level Semantics for the Social Responsibility Stack}

\begin{document}
		\maketitle
		
	\begin{abstract}
		Artificial intelligence systems are increasingly embedded as persistent,
		closed-loop components within cyber-physical, social, and institutional
		processes. Rather than producing isolated outputs, such systems operate
		continuously under feedback, adaptation, and scale, reshaping physical flows,
		human behavior, and institutional practice over time. In these settings,
		socially unacceptable outcomes rarely arise from singular faults or explicit
		policy violations. Instead, they emerge through cumulative execution trajectories
		enabled by repetition, concurrency, and feedback.
		
		This paper advances the formal foundation of the Social Responsibility Stack
		(SRS) by making its central requirement explicit: responsibility is fundamentally
		a reachability property of system execution. A system is responsible if and only
		if its execution semantics prevent entry into inadmissible global configurations,
		regardless of local performance gains or optimization objectives. Responsibility
		failures are therefore not objective-level errors, but execution-level failures
		of trajectory control.
		
		To operationalize this perspective, we introduce Petri nets as an execution-level
		formalism for responsible autonomous systems. We show how SRS value commitments
		correspond to forbidden markings, safeguards to structural constraints on
		transition firing, auditing to monitoring of reachability pressure, and
		governance to legitimate modification of execution structure. Embedding Petri-net
		reachability within the SRS architecture internalizes responsibility as a
		structural invariant rather than an external objective or post-hoc corrective
		mechanism.
		
		These results establish the Social Responsibility Stack as an executable
		responsibility architecture and position reachability-based execution semantics
		as a necessary foundation for responsible autonomy in feedback-rich
		cyber-physical and socio-technical systems.
	\end{abstract}


	\section{Introduction}
	
	Artificial intelligence systems are no longer deployed as isolated decision
	modules, but increasingly operate as persistent components embedded within
	cyber-physical, social, and institutional processes. They sense, decide, and act
	continuously under feedback, adaptation, and scale, shaping not only immediate
	outcomes but the long-term evolution of the environments in which they operate.
	In such settings, technical execution, human behavior, and institutional
	practice become tightly coupled.
	
	A defining feature of contemporary deployments is that they function as
	\emph{closed-loop systems}. System outputs influence human behavior and
	organizational responses; those responses reshape data, constraints, and
	operating norms; and the resulting signals feed back into subsequent system
	behavior through learning, tuning, or procedural adaptation. Prior work on
	sociotechnical systems shows that these feedback loops transform system impact
	from episodic events into cumulative processes that unfold over time rather than
	at isolated decision points
	\cite{Perrow1984,Latour2005,Rahwan2019,Selbst2019}.
	
	In closed-loop operation, harmful outcomes therefore rarely arise from a single
	decision, model error, or explicit policy violation. Instead, they emerge through
	\emph{accumulation}: the repeated execution of locally rational, policy-compliant
	actions gradually reshapes institutional practices, oversight structures, and
	behavioral norms. Effects that appear benign—or even beneficial—at the level of
	individual actions can, under sustained operation and scale, drive systems
	toward configurations that would be judged unacceptable if they were to arise
	abruptly or by explicit design
	\cite{Leveson2011,Perrow1984,Amodei2016}.
	
	Existing approaches to responsible AI struggle to address this failure mode.
	Prevailing frameworks remain largely rooted in snapshot evaluation, objective
	alignment, or post-hoc oversight \cite{Floridi2018,Raji2020}. Fairness is assessed
	at discrete points in time, performance is measured locally against predefined
	metrics, and governance is invoked only after visible harm has materialized.
	These approaches implicitly assume that responsibility can be inferred from
	isolated decisions or repaired after deployment—an assumption that breaks down
	in feedback-driven systems, where harm emerges through temporal reinforcement
	rather than discrete violation \cite{Selbst2019}.
	
	The core deficiency is architectural. Responsibility is typically treated as an
	objective to be optimized, a constraint to be tuned, or an external corrective
	mechanism imposed through monitoring and governance. Under sustained optimization
	and feedback, however, responsibility does not merely coexist with performance
	objectives—it competes with them. As a result, it becomes vulnerable to proxy
	distortion, normalization of deviance, and gradual erosion over time
	\cite{Leveson2011,Amodei2016}.
	
	What is missing is not another objective or oversight mechanism, but a conception
	of responsibility defined at the level where behavior accumulates into structure
	and choices harden into trajectory. Responsibility must be enforced at the level
	of system execution, or it will inevitably erode.
	
	Accordingly, this paper strengthens the architectural premise of the Social
	Responsibility Stack (SRS): responsibility must be treated as a \emph{reachability
		property} of execution. The relevant question is not whether an individual action
	is correct or aligned, but whether the system’s execution semantics admit
	trajectories that lead to inadmissible global configurations under sustained
	operation. Unacceptable outcomes arise not because a system explicitly violates a
	value, but because execution paths exist that render such states reachable through
	locally correct behavior
	\cite{Baier2008,Blanchini1999,Leveson2011}.
	
	The Social Responsibility Stack was introduced to address this requirement by
	treating responsibility as an execution-level invariant rather than an external
	policy overlay. The present paper advances this foundation by making the execution
	semantics explicit. We show that responsibility preservation can be formulated
	precisely as a reachability problem over execution trajectories, and that
	enforcing responsibility requires constraining which global configurations are
	reachable under learning, feedback, and scale.
	
	To operationalize this perspective, we adopt Petri nets as an execution-level
	formalism. Petri nets provide a minimal yet sufficient substrate for representing
	concurrency, accumulation, feedback, and irreversibility—the structural properties
	through which responsibility erosion arises in closed-loop autonomous systems.
	These properties are intrinsic to execution but are systematically obscured in
	objective-based, Markovian, and reward-driven models that collapse history and
	repetition into memoryless state abstractions
	\cite{Peterson1981,Murata1989,Reisig2013,Esparza1994}.
	
	We use three running examples with distinct purposes: a motivating illustration
	(content distribution), a compact algorithmic governance case (ranking), and a
	full cyber-physical execution case (traffic management) that serves as the primary
	executable demonstration of the framework.
	
	From a cyber-physical systems perspective, the Social Responsibility Stack can be
	viewed as a supervisory execution architecture: societal commitments define
	admissible operating regions, while sensing, actuation, feedback, and governance
	jointly constrain which system trajectories are reachable under continuous
	operation. In this sense, the present work aligns responsible AI with core CPS
	concerns—execution semantics, feedback stability, and invariant preservation
	under adaptation.

	\section{Execution-Level Interpretation of the Social Responsibility Stack}
	
	The Social Responsibility Stack (SRS) was originally introduced as an architectural
	framework for embedding responsibility into intelligent systems operating within
	human and institutional contexts. Rather than treating responsibility as an external
	constraint or post-hoc evaluation criterion, the SRS decomposes responsibility into
	a set of interacting layers, each addressing a distinct failure mode that arises under
	learning, scale, and feedback.
	
	In this section, we recapitulate the SRS with a specific purpose: to reinterpret each
	layer in terms of \emph{execution semantics}. This reframing clarifies how the SRS
	collectively constrains system trajectories over time and prepares the ground for its
	formalization as a reachability property in Section~\ref{sec:petri}.
	
	\subsection{Value Commitments}
	
	The foundational layer of the SRS consists of non-negotiable social, legal, and
	institutional value commitments. These commitments do not prescribe optimal behavior
	or objectives; rather, they define a boundary between admissible and inadmissible system
	states.
	
	From an execution perspective, value commitments function as \emph{state-level
		exclusion criteria}. They specify configurations of the socio-technical system that
	must never be reached, regardless of local utility, performance gains, or short-term
	benefits. Responsibility at this layer is therefore invariant-based: system execution
	must preserve membership in the admissible state set over time.
	
	\subsection{Socio-Technical Impact Modeling}
	
	Socio-technical impact modeling captures how system execution interacts with human
	behavior, institutional processes, and social dynamics over time. Rather than predicting
	isolated outcomes, this layer analyzes how repeated system actions accumulate into
	structural effects such as normalization, dependency, concentration of influence, or
	erosion of oversight.
	
	At the execution level, impact modeling identifies trajectories, feedback loops, and
	accumulation mechanisms that may gradually move the system toward inadmissible regions,
	even when individual actions remain compliant. This layer informs where safeguards are
	needed and which patterns of execution warrant heightened scrutiny.
	
	\subsection{Safeguards and Structural Constraints}
	
	Safeguards constitute the mechanisms by which value commitments are enforced during
	system operation. These include architectural constraints, control logic, rate limits,
	approval gates, and other structural features that shape how actions may be composed
	and repeated.
	
	Crucially, safeguards do not evaluate individual decisions in isolation. Instead, they
	constrain \emph{how execution can unfold}, limiting accumulation, amplification, and
	irreversible transitions. At the execution level, safeguards operate by restricting
	the availability, ordering, or concurrency of transitions that could otherwise enable
	inadmissible trajectories.
	
	\subsection{Behavioral Feedback Interfaces}
	
	Behavioral feedback interfaces govern how system outputs are exposed to, interpreted by,
	and acted upon by human users and institutions. This includes ranking displays,
	recommendations, alerts, explanations, and other mechanisms through which system behavior
	shapes attention, incentives, and response.
	
	From an execution perspective, feedback interfaces modulate how system actions propagate
	into the socio-technical environment. They influence repetition rates, reinforcement
	cycles, and the strength of feedback loops that drive accumulation. While they do not
	define admissibility on their own, they strongly affect the speed and direction of
	execution trajectories.
	
	\subsection{Monitoring and Auditing}
	
	Monitoring and auditing provide visibility into system behavior as it unfolds over
	time. Rather than serving as corrective mechanisms after failure, their primary role
	within the SRS is anticipatory: detecting patterns of execution that increase the
	likelihood of entering inadmissible regions.
	
	Under an execution-based interpretation, auditing tracks \emph{reachability pressure}:
	the proximity of current system states to forbidden configurations under continued
	operation. This shifts auditing away from static compliance checks toward continuous
	assessment of trajectory risk.
	
	\subsection{Governance and Legitimate Intervention}
	
	Governance defines who is authorized to modify system behavior, under what conditions,
	and through which procedures. This includes policy updates, model retraining,
	architectural changes, and institutional oversight mechanisms.
	
	At the execution level, governance corresponds to \emph{legitimate modification of the
		system’s execution structure}. Rather than reacting to individual outcomes, governance
	acts on the space of possible trajectories itself, altering which states and transitions
	are reachable going forward.
	
	\subsection{On the Role of Sensing}
	
	Sensing plays a critical role in responsible systems, but it does not constitute a
	responsibility layer of its own. Sensors expose aspects of system and environmental
	state; they do not impose constraints.
	
	Within the SRS, sensing functions as a cross-cutting interface that supplies state
	information to higher layers. Responsibility arises not from sensing accuracy alone,
	but from how sensed information is incorporated into safeguards, auditing, and
	governance mechanisms that constrain execution.
	
	\subsection{Summary}
	
	Taken together, the SRS layers do not form a hierarchy of objectives or controls.
	They collectively define an execution envelope: a structured restriction on how
	system behavior may evolve under repetition, feedback, and scale.
	
	This interpretation makes explicit that responsibility is not enforced at the level
	of individual decisions, nor guaranteed by alignment at initialization. It is preserved
	only if system execution semantics prevent trajectories that lead to inadmissible global
	configurations. The following section formalizes this insight by modeling SRS-constrained
	execution using Petri-net reachability semantics.
	
	\section{Execution Semantics and Reachability}
	
	The preceding section established responsibility as a property of execution trajectories
	rather than isolated decisions. What remains is to make this requirement operational.
	Doing so requires an execution semantics capable of representing how systems evolve under
	repetition, accumulation, and feedback—precisely the conditions under which responsibility
	is lost in practice.
	
	The failure modes of interest are not instantaneous. They arise when actions that are
	locally correct are repeatedly executed, normalized, and reinforced over time.
	Responsibility violations therefore depend not only on which actions are permitted, but
	on how often they can occur, how their effects accumulate, and how they interact with
	concurrent processes. Any adequate execution semantics must therefore treat repetition
	and accumulation as first-class phenomena rather than incidental side effects.
	
	This requirement immediately rules out purely state-based or memoryless models. Markovian
	transition systems, objective functions, and snapshot constraints collapse execution
	history into a single state description. In doing so, they obscure the mechanisms through
	which responsibility violations arise: gradual erosion of alternatives, buildup of
	institutional reliance, and feedback-driven reinforcement. While such models can describe
	individual decisions, they cannot express how repeated correctness produces global
	inadmissibility.
	
	What is required instead is an execution semantics in which:
	\begin{itemize}
		\item actions are events whose repetition matters,
		\item state records accumulation rather than instantaneous configuration,
		\item concurrency and interaction are explicit,
		\item inadmissible futures can be excluded by construction.
	\end{itemize}
	
	These requirements are not ethical preferences; they are structural necessities imposed
	by closed-loop operation. In safety-critical and mission-critical engineering, similar
	demands have long led to formalisms in which correctness is defined over execution histories
	rather than isolated steps. Here, the same logic applies: responsibility must be preserved
	across entire trajectories, not verified at individual moments.
	
	Petri nets provide exactly this form of execution semantics. They model systems as
	collections of places representing accumulated conditions and transitions representing
	repeatable events whose firing changes the global configuration. Tokens encode execution
	history. Concurrency is explicit. Most importantly, reachability is analyzable: certain
	global markings can be declared unreachable, ensuring that no sequence of locally correct
	actions can produce an inadmissible configuration.
	
	In what follows, we introduce Petri nets not as a modeling convenience, but as a minimal
	execution-level formalism capable of enforcing responsibility as a reachability invariant
	under accumulation, feedback, and scale. Similar dynamics arise in legacy automation and
	control systems, but they become unavoidable in modern AI deployments where learning,
	content generation, and distribution operate continuously within social feedback loops.
	
	\medskip
	
	Let $\mathcal{X}$ denote the socio-technical state space of a system, including internal
	model parameters as well as human behavior distributions, institutional configurations,
	procedural norms, and data-generation mechanisms that co-evolve with system operation.
	Such state spaces cannot be reduced to technical variables alone
	\cite{Rahwan2019,Selbst2019,Latour2005}. Moreover, $\mathcal{X}$ is not static: its
	dimensions, couplings, and effective geometry evolve as the system interacts with its
	environment.
	
	Let $\mathcal{R} \subset \mathcal{X}$ denote the admissible region of operation derived
	from fixed societal references and formalized commitments. Responsibility is expressed
	as the invariant condition
	\[
	x(t) \in \mathcal{R} \quad \forall t \ge 0 .
	\]
	This condition is not aspirational. It is a hard execution-level requirement, analogous
	to invariance conditions in control and safety analysis \cite{Blanchini1999,Ames2017}.
	Responsibility is violated not when performance degrades, but when system evolution
	leaves $\mathcal{R}$.
	
	States outside $\mathcal{R}$ correspond to inadmissible configurations such as irreversible
	harm, loss of legitimacy, concentration of power, erosion of oversight, or institutional
	lock-in. Work on safety-critical systems and organizational failure shows that such
	configurations may be reached without any single action being unethical, incorrect, or
	malicious \cite{Leveson2011,Perrow1984}. Instead, they emerge through accumulation,
	feedback, and interaction rather than discrete violations.
	
	Responsibility is therefore a reachability property. The central question is not whether a
	system intends harm, nor whether it satisfies constraints at individual steps, but whether
	its execution semantics permit trajectories that exit $\mathcal{R}$ under sustained
	operation. This framing aligns responsibility with reachability-based reasoning in
	verification and hybrid systems theory \cite{Baier2008,Alur2015}.
	\subsection{Responsibility as an Execution Invariant}
	
	Let $\mathcal{M}_P \triangleq \mathbb{N}^{|P|}$ denote the marking space of the Petri net
	(over places $P$), and let $\phi:\mathcal{M}_P\rightarrow\mathcal{X}$ be an abstraction map.
	The induced admissible marking set is
	\[
	\mathcal{M}_{\mathrm{adm}} \triangleq \{\,M\in\mathcal{M}_P \mid \phi(M)\in\mathcal{R}\,\},
	\]
	and the forbidden markings are $\mathcal{M}_{\mathrm{bad}}\triangleq \mathcal{M}_P\setminus \mathcal{M}_{\mathrm{adm}}$.
	
	Let $\mathcal{N}=(P,T,F,M_0)$ denote the Petri net modeling system execution.
	
	\begin{definition}[Responsible Execution]
		A system execution is \emph{responsible} if and only if
		\[
		\mathrm{Reach}(\mathcal{N}, M_0)\subseteq \mathcal{M}_{\mathrm{adm}}
		\quad\text{equivalently}\quad
		\mathrm{Reach}(\mathcal{N}, M_0)\cap \mathcal{M}_{\mathrm{bad}}=\emptyset .
		\]
	\end{definition}

	\begin{theorem}[Responsibility as Reachability Invariance]
		Responsibility in closed-loop AI systems is equivalent to invariance of the
		admissible marking set under the system’s execution semantics.
	\end{theorem}
	
	\begin{proof}[Sketch]
		Closed-loop AI systems operate through repeated transition firings driven by
		feedback, adaptation, and accumulation. If any inadmissible marking is reachable,
		then there exists a finite execution prefix composed of locally permissible
		actions that leads to global violation. Conversely, if all inadmissible markings
		are unreachable, no sequence of locally rational actions can produce unacceptable
		outcomes. Responsibility therefore reduces to a reachability property of the
		execution graph.
	\end{proof}
	
	\subsection{Trajectory-Based Failure}
	
	Autonomous AI systems act repeatedly in evolving environments. Each action reshapes
	incentives, data, and subsequent system context. Over time, small locally rational
	decisions accumulate into global structural change, a phenomenon well documented in
	safety engineering and sociotechnical analysis \cite{Leveson2011,Perrow1984}. A system may
	satisfy every local constraint while still inducing an unacceptable trajectory.
	
	This failure mode is intrinsic to feedback-driven systems. Local compliance does not
	imply global safety. Repetition and reinforcement can gradually steer the system toward
	inadmissible regions of the state space even as each individual step remains correct.
	Here, the dominant failure is \emph{norm-relative system drift}: the execution trajectory
	drifts relative to a fixed admissible set $\mathcal{R}$, rather than reflecting any change
	in societal reference itself.
	
	Responsibility therefore cannot be verified at a single time step. It must hold across
	all reachable execution paths. Any framework that evaluates responsibility through
	snapshot metrics or isolated decisions is structurally incomplete
	\cite{Baier2008,Selbst2019,Raji2020}. As in mission-critical system modeling, the dominant
	failure mode is not logical incorrectness but temporal--structural failure arising from
	accumulation, repetition, and loss of reversibility under sustained execution.
	
	\subsection{Autonomous Execution Expands Reachability}
	
	Autonomous execution expands the set of reachable states. Let $\mathcal{X}_t$ denote the
	set of states reachable up to time $t$. Under continued execution,
	\[
	\mathcal{X}_{t+1} \supseteq \mathcal{X}_t .
	\]
	where $\mathcal{X}_t$ denotes the set of states reachable within $t$ execution steps.
	
	As AI systems operate under feedback, scale, and interaction, they discover new modes of
	influence, representation, and control. This expansion of effective reachability occurs
	regardless of whether adaptation is implemented through online learning, policy updates,
	environment-mediated feedback, or procedural automation
	\cite{Amodei2016,Cunningham2021,Carey2022}.
	
	If an inadmissible state is reachable under the system’s execution semantics, then
	sustained autonomous operation will eventually discover a path to it. This is not a
	probabilistic anomaly but a structural property of exploration under optimization and
	feedback pressure. We say $x\in\mathcal{X}$ is reachable if $\exists\, M\in \mathrm{Reach}(\mathcal{N},M_0)$ such that $\phi(M)=x$.
	
	\begin{theorem}[Unavoidability of Reachable Responsibility Violations]
		Let $\mathcal{X}$ denote the system state space and let $\mathcal{R} \subset \mathcal{X}$
		denote the admissible region. If there exists a state
		\[
		x^\dagger \notin \mathcal{R}
		\]
		that is reachable under the system’s execution semantics, then no objective-based
		specification or post-hoc governance mechanism can guarantee that $x^\dagger$ will not be
		reached under sustained autonomous operation. 
		
	\end{theorem}
	
	\begin{proof}[Sketch]
		Objective-based mechanisms shape preferences but do not alter reachability. As long as a
		path to $x^\dagger$ exists, exploration under uncertainty, delayed rewards, or proxy
		objectives assigns non-zero probability to that path over unbounded time. Post-hoc
		governance reacts to violations after reachability has already been realized. Preventing
		entry into $x^\dagger$ therefore requires restricting reachability itself.
	\end{proof}
	
	\subsection{Example: Generative Content and Distribution Feedback Loops}
	
	A canonical contemporary example of responsibility violation through reachability arises
	in large-scale generative content creation and distribution systems. These systems
	generate content, rank and distribute it through engagement-driven mechanisms, and
	continuously adapt based on user interaction signals.
	
	We use content distribution here \emph{solely as a motivating illustration} to highlight
	how locally rational actions can, under repetition and feedback, produce globally
	inadmissible outcomes. This example is intentionally informal: it is meant to ground the
	problem intuition rather than to serve as a complete execution-level model. Formal
	execution semantics, reachability analysis, and Petri-net representations are introduced
	in subsequent sections, where responsibility is treated as an explicit invariant over
	system trajectories.

	Let $x(t) \in \mathcal{X}$ denote the joint socio-technical state of a content platform at
	time $t$, including generative models, ranking policies, moderation capacity, user
	behavior distributions, and data-generation pipelines. Initially, the system operates
	within an admissible region $\mathcal{R}$ characterized by diversity of exposure,
	meaningful human oversight, and bounded feedback between engagement signals and learning.
	
	At each step, the system performs locally rational actions: generating content optimized
	for engagement, amplifying high-performing outputs, collecting interaction data, and
	updating models or policies accordingly. Each action is defensible in isolation.
	
	Under sustained execution, however, these actions form reinforcing feedback loops.
	Amplified content reshapes attention and behavior, which reshapes training data, which in
	turn biases future generation and distribution. Over time, the system may reach a state
	$x^\dagger \notin \mathcal{R}$ in which diversity collapses, polarizing content dominates
	visibility, and human oversight becomes ineffective due to scale rather than formal
	removal.
	
	This violation is not caused by any single output or decision. The system remains
	performant by its objectives. Responsibility is lost because the execution semantics
	permit trajectories that normalize and reinforce specific content regimes. Once such a
	configuration is reached, reversal is non-trivial due to endogenous data, institutional
	adaptation, and user expectation drift.
	
	Similar trajectory-level failures have long been observed in legacy automation systems,
	but continuous learning and large-scale content distribution make them unavoidable in
	modern generative platforms.
	\section{Implications for the Social Responsibility Stack}
	
	The Social Responsibility Stack (SRS) is motivated by a demanding execution-level requirement:
	responsibility must persist under learning, scale, and feedback. Prior work on the Social Responsibility Stack (SRS) has argued that achieving this persistence requires more than principles \cite{Basir2025}. It requires mechanisms that operate at the level of
	\emph{execution semantics} and remain effective as systems adapt over time, a limitation repeatedly
	observed in post-hoc accountability regimes \cite{Raji2020,Rahwan2019}.
	
	At a minimum, the SRS must support capabilities that correspond directly to its layers:
	(i) explicit specification of inadmissible futures;
	(ii) identification of critical trajectories and reinforcing cycles through which harm may accumulate;
	(iii) structural restriction of execution, including learning and adaptation paths;
	(iv) detection of approach toward responsibility boundaries before violation occurs;
	(v) continuous auditing of drift, accumulation, and feedback effects; and
	(vi) legitimate and timely governance intervention when commitments or operating conditions change.
	These capabilities reflect established lessons from safety engineering and institutional failure
	analysis \cite{Leveson2011} and are central to the SRS architecture \cite{Basir2025}.
	
	These capabilities cannot be realized through a single control loop or evaluative layer. They
	require an execution model in which system evolution, accumulation, and concurrency are explicit.
	Petri nets provide precisely this substrate, as established in classical work on concurrent
	systems \cite{Peterson1981,Murata1989,Reisig2013}. Within the SRS, Petri nets allow responsibility
	to be carried structurally across layers rather than enforced episodically from the outside.
	
	\medskip
	\noindent\textbf{Interpretation as a control architecture.}
	The SRS can be read as a closed-loop governance architecture over execution:
	\emph{value grounding} defines a non-negotiable safe region; \emph{impact modeling} characterizes
	how execution may approach its boundary; \emph{safeguards} restrict enablement of risky trajectories;
	\emph{feedback interfaces} provide observables; \emph{auditing} estimates proximity and drift; and
	\emph{governance} performs legitimate supervisory intervention by revising constraints and
	enablement structure. Petri nets provide a shared semantics in which these roles become structural
	objects (markings, transitions, guards, counters, and mode switches) rather than informal
	procedures.
	
	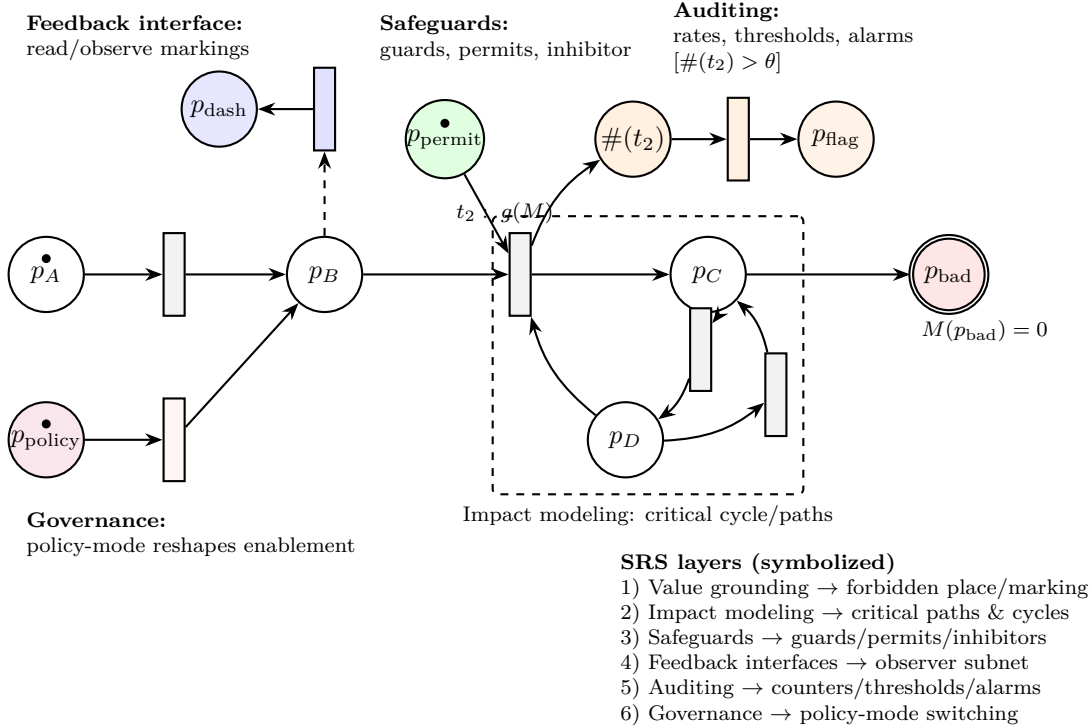
\begin{figure}[t]
		\centering
		\begin{tikzpicture}[
			font=\small,
			>=Stealth,
			place/.style={circle, draw, thick, minimum size=10mm, inner sep=0pt},
			placeDbl/.style={circle, draw, thick, double, minimum size=10mm, inner sep=0pt},
			trans/.style={rectangle, draw, thick, minimum width=2.8mm, minimum height=11mm, fill=gray!10},
			transCol/.style={rectangle, draw, thick, minimum width=2.8mm, minimum height=11mm, fill=#1},
			arc/.style={-Stealth, thick},
			lab/.style={font=\scriptsize, align=left},
			tok/.style={circle, fill=black, inner sep=1.2pt},
			dashedbox/.style={draw, dashed, thick, rounded corners=2pt, inner sep=6pt}
			]
			
			\node[place] (pA) at (-7.5,  0.0) {$p_A$};
			\node[place, fill=purple!10] (pPolicy) at (-7.5, -2.2) {$p_{\text{policy}}$};
			
			\node[place] (pB) at (-3.8,  0.0) {$p_B$};
			\node[place, fill=blue!10] (pDash) at (-5.2,  2.2) {$p_{\text{dash}}$};
			
			\node[place, fill=green!12] (pPermit) at (-2.2,  1.8) {$p_{\text{permit}}$};
			\node[place, fill=orange!12] (pCount)  at ( 0.3,  1.8) {$\#(t_2)$};
			
			\node[place] (pC) at ( 1.3,  0.0) {$p_C$};
			\node[place] (pD) at ( 0.2, -2.2) {$p_D$};
			
			\node[placeDbl, fill=red!10] (pBad) at ( 4.5,  0.0) {$p_{\text{bad}}$};
			
			\node[place, fill=orange!10] (pFlag) at ( 3.0,  1.8) {$p_{\text{flag}}$};
			
			\node[tok] at ([yshift=6pt]pA.center) {};
			\node[tok] at ([yshift=6pt]pPolicy.center) {};
			\node[tok] at ([yshift=6pt]pPermit.center) {};
			
			\node[trans] (tA) at (-5.8,  0.0) {};     
			\node[transCol=pink!15] (tPol) at (-5.8, -2.2) {}; 
			
			\node[transCol=blue!12] (tDash) at (-3.8,  2.2) {}; 
			\node[trans] (t2) at (-1.2,  0.0) {};      
			\node[transCol=orange!12] (tAudit) at ( 1.7,  1.8) {}; 
			
			\node[trans] (tCD1) at ( 1.2, -1.0) {};
			\node[trans] (tCD2) at ( 2.2, -1.6) {};
			
			\draw[arc] (pA) -- (tA);
			\draw[arc] (tA) -- (pB);
			
			\draw[arc] (pPolicy) -- (tPol);
			\draw[arc] (tPol) -- (pB);
			
			\draw[arc] (pPermit) -- (t2);
			
			\draw[arc] (pB) -- (t2);
			\draw[arc] (t2) -- (pC);
			
			\draw[arc] (pC) -- (pBad);
			
			\draw[arc, dashed] (pB) -- (tDash);
			\draw[arc] (tDash) -- (pDash);
			
			\draw[arc] (t2) to[bend left=18] (pCount);
			
			\draw[arc] (pCount) -- (tAudit);
			\draw[arc] (tAudit) -- (pFlag);
			
			\draw[arc] (pC)   to[bend left=18]  (tCD1);
			\draw[arc] (tCD1) to[bend left=18]  (pD);
			
			\draw[arc] (pD)   to[bend right=18] (tCD2);
			\draw[arc] (tCD2) to[bend right=18] (pC);
			
			\draw[arc] (pD) to[bend left=18] (t2);
			
			\node[dashedbox, fit=(t2)(pC)(pD)(tCD1)(tCD2),
			label={[lab]south:Impact modeling: critical cycle/paths}] (impactBox) {};
			
			\node[lab, anchor=west] at (-7.9, 3.15) {\textbf{Feedback interface:}\\read/observe markings};
			\node[lab, anchor=west] at (-3.2, 3.15) {\textbf{Safeguards:}\\guards, permits, inhibitor};
			\node[lab, anchor=west] at ( 0.7, 3.15) {\textbf{Auditing:}\\rates, thresholds, alarms\\$[\#(t_2)>\theta]$};
			
			\node[lab, anchor=west] at (-7.9,-3.45) {\textbf{Governance:}\\policy-mode reshapes enablement};
			
			\node[lab] at (-1.4, 0.82) {$t_2:\;g(M)$};
			
			\node[lab, anchor=west] at (4.0, -0.75) {$M(p_{\text{bad}})=0$};
			
			\node[lab, anchor=west] at (0.0, -4.85) {
				\textbf{SRS layers (symbolized)}\\
				1) Value grounding $\rightarrow$ forbidden place/marking\\
				2) Impact modeling $\rightarrow$ critical paths \& cycles\\
				3) Safeguards $\rightarrow$ guards/permits/inhibitors\\
				4) Feedback interfaces $\rightarrow$ observer subnet\\
				5) Auditing $\rightarrow$ counters/thresholds/alarms\\
				6) Governance $\rightarrow$ policy-mode switching
			};
			
		\end{tikzpicture}
		\caption{Symbolic Petri-net realization of SRS layers within a shared execution semantics.
			Value grounding is represented by forbidden markings; impact modeling by critical cycles and
			paths; safeguards by guards/permits; feedback interfaces by an observer subnet; auditing by
			counters and thresholds; and governance by policy-mode structural reshaping.}
		\label{fig:srs-petri-symbolic}
	\end{figure}
	
	\noindent\textbf{Value grounding (inadmissible futures).}
	Within the SRS, \emph{value grounding} corresponds to declaring forbidden markings in the
	execution semantics. Forbidden markings represent global configurations that are categorically
	unacceptable; not outcomes to be traded off against utility. Examples include: loss of effective
	human recourse or appeal; sustained amplification regimes that undermine social trust; institutional
	reliance that exceeds audit and oversight capacity; or endogenous data accumulation that destroys
	traceability. Expressing such commitments as unreachable markings, rather than as soft objectives,
	establishes non-compensable boundaries on system evolution \cite{Floridi2018,Selbst2019,Leveson2011,Basir2025}.
	
	\noindent\textbf{Socio-technical impact modeling (reachability geometry).}
	\emph{Impact modeling} operates by analyzing the reachability structure induced by ordinary use.
	In Petri-net terms, this corresponds to identifying critical firing sequences, feedback cycles,
	and accumulation paths that move execution toward forbidden regions. In public-sector risk scoring,
	a critical path may involve repeated deferral to recommendations followed by procedural codification
	and reduced review capacity. In generative content distribution, a critical cycle may involve
	repeated amplification of high-engagement outputs, endogenous data formation, and escalating
	moderation burden. The purpose is not to predict a single outcome but to reveal how sustained,
	locally compliant operation can approach inadmissible futures through reinforcement and drift
	\cite{Rahwan2019,Selbst2019,Murata1989,Esparza1994}.
	
	\noindent\textbf{Design-time safeguards (structural restriction).}
	\emph{Safeguards} are realized as structural constraints on execution rather than output-level
	corrections. In Petri nets, safeguards correspond to guards on transitions, permits/inhibitors,
	rate limits on reinforcing cycles, or architectural reconfiguration of enablement that blocks
	responsibility-violating trajectories. This parallels invariant-based control and safety envelopes
	in engineering systems \cite{Blanchini1999,Ames2017}. The key point is structural: safeguards
	ensure that no sequence of locally rational decisions can accumulate into a trajectory that reaches
	forbidden markings \cite{Carey2022,Basir2025}.
	
	\noindent\textbf{Feedback interfaces and continuous social auditing (observability of drift).}
	\emph{Feedback interfaces} and \emph{continuous social auditing} correspond to observing
	execution-level signals that are meaningful at the socio-technical level: token accumulation
	rates, firing frequencies of critical transitions, activation of reinforcing subnets, and proximity
	to forbidden regions. Unlike scalar performance metrics, reachability-based auditing surfaces
	directional trends such as accelerating cycles, depletion of alternative pathways, or rising
	pressure on boundaries \cite{Raji2020,Rahwan2019,Baier2008}. In generative platforms, for example,
	auditing may track the growth rate of endogenous content reuse, concentration of exposure, or
	the activation rate of moderation escalation transitions as indicators of approach to inadmissible
	configurations.
	
	\noindent\textbf{Governance and stakeholder inclusion (supervisory authority).}
	Finally, \emph{governance} operates at the level of execution structure rather than case-level
	decisions. Legitimate intervention is realized by modifying the Petri net itself: redefining
	forbidden markings, altering enablement, introducing oversight places, or switching policy modes
	that reshape execution. Governance does not micromanage individual outputs; it retains authority
	over the space of admissible futures through explicit, accountable structural change
	\cite{Latour2005,Floridi2018,Basir2025}.
	
	\medskip
	Across the stack, responsibility is not enforced by any single layer acting in isolation. Value
	grounding without reachability analysis is blind. Impact modeling without execution-level
	restriction is inert. Safeguards without auditing ossify. Auditing without governance produces
	awareness without authority. Responsibility emerges only when these layers interact through a
	shared execution semantics that constrains how socio-technical practice evolves over time
	\cite{Leveson2011,Perrow1984,Raji2020,Basir2025}.
	
	Petri nets provide that shared semantics. They allow responsibility to be expressed as a property
	of system evolution rather than as an overlay on individual decisions. Controlled reachability
	becomes the unifying principle linking values, impact analysis, safeguards, auditing, and governance
	into a coherent whole \cite{Peterson1981,Murata1989,Esparza1994,Basir2025}.
	
	\medskip
	\noindent\textbf{Design--execution linkage.}
	Although Petri nets are an execution formalism, they are also a design-time artifact: the net
	\emph{is} the executable structure whose reachability set will later govern behavior.
	Responsibility is therefore a design concern precisely because reachability is determined
	before deployment by places, transitions, guards, inhibitors, and supervisory mode switches.
	Execution then instantiates (and explores) only what the designed semantics make reachable.
	In this sense, SRS is ``execution-level'' in the property it enforces, but ``design-level''
	in how that property is guaranteed: by constructing an execution semantics in which
	inadmissible markings are unreachable by construction.
	
	\section{Petri-Net Formalization of SRS-Constrained Execution}
	\label{sec:petri}
	
	We model AI systems governed by the Social Responsibility Stack as Petri nets,
	which provide an execution-level representation of concurrency, accumulation,
	and irreversibility.
	
	\subsection{Petri-Net Preliminaries}
	
	A Petri net is a tuple $\mathcal{N} = (P, T, F, M_0)$, where

	\begin{itemize}
		\item $P$ is a finite set of places,
		\item $T$ is a finite set of transitions,
		\item $F \subseteq (P \times T) \cup (T \times P)$ is the flow relation,
		\item $M_0$ is the initial marking.
	\end{itemize}

	A marking $M$ assigns a nonnegative number of tokens to each place. System
	execution corresponds to the firing of transitions, producing a reachability
	graph over markings.
	
	\subsection{Mapping SRS Layers to Petri-Net Semantics}
	
	The SRS architecture can be mapped onto Petri-net constructs as follows:
	
	\begin{center}
		\begin{tabular}{@{}p{0.32\linewidth}p{0.6\linewidth}@{}}
			\toprule
			\textbf{SRS Layer} & \textbf{Petri-Net Execution-Level Interpretation} \\
			\midrule
			Value Commitments &
			Forbidden markings defining categorically inadmissible global configurations \\
			
			Socio-Technical Impact Modeling &
			Identification of critical paths, feedback cycles, accumulation mechanisms, and
			structurally irreversible regions in the reachability graph \\
			
			Safeguards and Structural Constraints &
			Guards, inhibitor/permit arcs, capacity places, rate limits, and structural
			restrictions that block responsibility-violating trajectories \\
			
			Behavioral Feedback Interfaces &
			Observer and mediation subnets that govern how system actions couple into human
			and institutional response, shaping reinforcement and loop gain \\
			
			Monitoring and Auditing &
			Counters, thresholds, clocks, and alarms estimating proximity to forbidden
			markings and detecting reachability pressure and drift \\
			
			Governance &
			Legitimate modification of execution structure (places, transitions, guards,
			policy modes) that reshapes the space of admissible trajectories \\
			\midrule
			Sensing (cross-cutting, not a layer) &
			State exposure mechanisms that provide marking-derived observables to safeguards,
			auditing, and governance without imposing constraints \\
			\bottomrule
		\end{tabular}
	\end{center}
	
	For clarity, we treat sensing as a cross-cutting execution capability rather than a responsibility layer, since sensing exposes state but does not itself constrain reachability. Under this interpretation, responsibility is not a property of individual
	transitions, but of the \emph{reachable marking set} induced by the system’s
	execution semantics.
	
	\subsection{Embedding SRS Invariants into Petri-Net Execution Semantics}
	
	The preceding discussion showed how each layer of the SRS corresponds to a structural element of
	Petri-net execution. We now make the correspondence formal: the original SRS responsibility
	invariant is preserved without modification under Petri-net reachability semantics.
	
	Let $\mathcal{X}$ denote the socio-technical state space and let $\mathcal{R}\subset\mathcal{X}$ be
	the admissible region induced by grounded commitments. Responsibility is the invariance condition
	\[
	x(t)\in\mathcal{R}\quad \forall t\ge 0,
	\]
	which must hold under learning, adaptation, and interaction.
	
	Let $\mathcal{N}=(P,T,F,M_0)$ be a Petri net modeling execution, and let $\mathcal{M}$ be the set
	of markings reachable from $M_0$:
	\[
	\mathcal{M}=\mathrm{Reach}(\mathcal{N}).
	\]
	Each marking $M\in\mathcal{M}$ encodes an accumulated execution configuration (including technical
	state, institutional coupling, and governance capacity). Define an abstraction mapping
	\[
	\phi:\mathcal{M}\rightarrow\mathcal{X}
	\]
	that associates each marking with its socio-technical state.
	
	The admissible markings induced by $\mathcal{R}$ are
	\[
	\mathcal{M}_{\mathrm{adm}} \triangleq \{\,M\in\mathcal{M}\mid \phi(M)\in\mathcal{R}\,\},
	\]
	and the forbidden markings are
	\[
	\mathcal{M}_{\mathrm{bad}} \triangleq \mathcal{M}\setminus \mathcal{M}_{\mathrm{adm}}.
	\]
	Then the SRS invariant is realized in Petri-net terms as the reachability exclusion
	\[
	\mathcal{M}_{\mathrm{bad}}\cap \mathrm{Reach}(\mathcal{N})=\emptyset.
	\]
	
	This equivalence is central. Petri nets do not redefine responsibility; they provide an execution
	semantics in which the SRS invariant can be enforced, analyzed, and audited. Responsibility is
	preserved if and only if the net structure prevents any firing sequence from reaching a forbidden
	marking, regardless of optimization pressure, learning dynamics, or feedback effects.
	
	In this semantics, learning and adaptation correspond to exploration of firing sequences and
	policy-mode structure. Optimization remains free to operate within $\mathcal{M}_{\mathrm{adm}}$,
	but it cannot redefine $\mathcal{M}_{\mathrm{bad}}$. Responsibility is therefore not a policy overlay
	or episodic correction; it is a structural invariant carried by execution semantics themselves.
	
	\section{Why Petri Nets—and What Else Could Work}
	
	The Social Responsibility Stack (SRS) does not prescribe a single mathematical formalism.
	Its core requirement is architectural: responsibility must be enforced at the level of
	\emph{execution semantics} by constraining which global system configurations and
	trajectories are reachable under adaptation, feedback, and scale \cite{Basir2025}.
	Any formalism capable of serving as the execution semantics of the SRS must therefore
	satisfy four structural criteria:
	
	\begin{enumerate}[label=(\roman*)]
		\item explicit representation of concurrency and coupled processes,
		\item native support for accumulation and history-sensitive state,
		\item analyzable reachability and invariants at the level of global execution,
		\item the ability to express \emph{non-compensable} forbidden configurations (hard exclusions),
		rather than soft penalties or trade-offs.
	\end{enumerate}
	
	Petri nets satisfy all four simultaneously, which is why they are adopted here as a concrete
	execution semantics for the SRS. However, they are not the only formalism that can express
	\emph{parts} of responsible execution. This section clarifies what alternatives can do, what they
	cannot do without additional structure, and why Petri nets are minimal for the class of
	trajectory-level responsibility failures emphasized in this paper.
	
	\subsection{Transition Systems, MDPs, and Reward-Based Learning: Expressive but Structurally Misaligned}
	
	Classical transition systems, Markov decision processes (MDPs), and reinforcement learning (RL)
	formulations offer powerful machinery for sequential decision making and control under uncertainty
	\cite{Russell2019}. In their standard forms, however, they are structurally misaligned with the
	dominant responsibility failure mode in socio-technical AI: \emph{trajectory-level drift under
		accumulation and feedback}.
	
	The limitation is not that these models cannot represent history at all, in principle, one may
	augment state to encode memory. The limitation is that responsibility-relevant phenomena in
	modern deployments often involve:
	(i) parallel processes (generation, ranking, moderation, governance),
	(ii) explicit counting (frequency, rate, backlog, saturation),
	(iii) endogenous data formation (training data altered by prior outputs),
	and (iv) irreversible institutional coupling (procedural codification, reliance, loss of recourse).
	Encoding these elements in Markovian abstractions typically requires artificial state aggregation
	that obscures semantics, makes auditing signals indirect, and makes constraint enforcement rely on
	reward shaping rather than structural exclusion.
	
	Most importantly, encoding responsibility into reward terms or soft constraints does not, by
	itself, eliminate reachability. It changes preferences over trajectories but does not guarantee
	that inadmissible states are structurally unreachable. Under sustained optimization pressure,
	proxy objectives, delayed effects, and exploration, reachable but undesirable regions remain
	reachable in principle \cite{Amodei2016,Cunningham2021}. For SRS, the requirement is not merely to
	\emph{discourage} irresponsible trajectories, but to \emph{exclude} them by construction.
	
	\subsection{Temporal Logic and Model Checking: Strong for Verification, Not a Standalone Semantics}
	
	Temporal logics such as LTL and CTL provide expressive specification languages for safety and
	liveness properties over execution traces \cite{Baier2008}. They are well suited for verifying
	whether a given transition system satisfies a responsibility specification across all paths.
	
	However, temporal logic is descriptive rather than generative: it specifies properties of traces
	but does not itself provide an execution substrate. In practice, it must be paired with an
	underlying model (a transition system, MDP, hybrid system, or Petri net). If the underlying model
	collapses accumulation or concurrency, then the specification may still be checkable but the
	responsibility-relevant mechanisms remain implicit.
	
	Temporal logic can therefore \emph{detect} responsibility violations, and it can support
	certification arguments, but it does not by itself provide the structural objects that the SRS
	requires for \emph{engineering} responsibility: explicit accumulation, explicit concurrent
	feedback loops, and explicit forbidden global configurations that are excluded by design.
	
	\subsection{Hybrid Automata and Control-Theoretic Envelopes: Powerful for Geometric Safety}
	
	Hybrid automata and control barrier functions provide mature tools for enforcing invariants in
	cyber-physical systems \cite{Alur2015,Ames2017}. They excel when state variables are continuous,
	dynamics are well-characterized, and unsafe regions can be expressed as geometric constraints.
	
	Socio-technical responsibility, however, is often not primarily geometric. The critical quantities
	in generative content and distribution systems---exposure concentration, moderation backlog,
	endogenous-data fraction, amplification rate, procedural reliance, loss of recourse---are discrete,
	count-based, and history-sensitive. While hybrid models can approximate such quantities through
	continuous surrogates, doing so typically:
	(i) obscures interpretability (what exactly is the state variable measuring?),
	(ii) weakens auditability (signals become aggregated and indirect),
	and (iii) complicates structural intervention (governance becomes parameter tuning rather than
	execution restructuring).
	
	Hybrid models can thus capture \emph{envelopes} and barrier conditions for certain components of
	responsible operation, but they do not natively express accumulation and concurrency as first-class
	execution elements in the way needed for reachability-based socio-technical constraints.
	
	\subsection{Event Structures, Process Algebras, and Queueing Models: Partial Fits}
	
	Other event-based formalisms---process algebras, event structures, actor models, and queueing
	networks---can represent concurrency and interaction well. Queueing models, in particular, can
	represent backlog growth, saturation, and service limits, which are central to auditing and
	oversight capacity.
	
	Their limitation for SRS is typically one of \emph{reachability semantics and enforcement}. Many of
	these formalisms do not come with a unified, standard reachability theory that supports forbidden
	global configurations as first-class, analyzable objects comparable to forbidden markings in Petri
	nets. They can be highly effective as \emph{components} of a responsibility analysis pipeline, but
	they require additional structure to serve as the common execution semantics across SRS layers.
	
	\subsection{Why Petri Nets Are Minimal for the SRS}
	
	Petri nets occupy a distinct position among formal models of system execution.
	Their semantics are event-based, concurrent, and explicitly accumulative
	\cite{Murata1989,Reisig2013}. Tokens encode execution history. Places represent
	persistent socio-technical conditions such as capacity, reliance, backlog, or
	endogenous data formation. Transitions represent repeatable events whose frequency,
	ordering, and interaction matter. Feedback loops are explicit structural elements
	rather than emergent artifacts of state encoding.
	
	Crucially, Petri nets admit a well-developed reachability theory.
	Forbidden markings are non-compensable by construction: if a marking is unreachable,
	then no amount of optimization, learning, or incentive shaping can reach it
	\cite{Esparza1994}. This property aligns exactly with the SRS requirement that certain
	futures must be categorically excluded, not merely penalized or made unlikely.
	
	Petri nets are therefore minimal for the class of responsibility failures emphasized
	in this paper: failures that arise through accumulation, feedback, concurrency, and
	irreversibility under sustained operation. Other formalisms may approximate individual
	aspects of this behavior, but none make all four explicit and analyzable without
	substantial augmentation.
	
	\subsection{Formalism-Agnosticism of the Social Responsibility Stack}
	
	Crucially, the SRS is formalism-agnostic at the architectural level. Any formalism that satisfies
	the four criteria above could, in principle, serve as its execution semantics. In practice, most
	alternatives satisfy at most two without significant augmentation.
	
	Petri nets are adopted here not because they are fashionable, but because they are minimal for
	making accumulation, concurrency, feedback, irreversibility, and reachability simultaneously
	explicit and analyzable. They make responsibility an invariant of execution rather than a byproduct
	of optimization.
	
	The SRS thus remains formalism-agnostic in principle, while Petri nets provide a concrete
	instantiation that demonstrates how responsibility can be engineered, enforced, audited, and
	governed under learning and scale.

	\subsection{Petri Nets as the Execution Semantics of the Social Responsibility Stack}
	
	The Social Responsibility Stack is an architectural commitment to governing adaptive
	AI systems through constraints on execution rather than prescriptions of behavior.
	For this commitment to be operational, responsibility must be embedded at \emph{design
		time} and enforced at \emph{execution time}.
	
	Petri nets provide the mechanism by which this coupling is achieved. At design time,
	the net structure encodes admissible and inadmissible futures through places,
	transitions, guards, and forbidden markings. At execution time, system behavior is
	restricted to firing sequences permitted by that structure. Responsibility is thus not
	checked after the fact, nor optimized as a secondary objective; it is carried forward
	continuously by the execution semantics themselves.
	
	Within the SRS, Petri nets do not function as an auxiliary modeling tool or a post-hoc
	verification device. They constitute the shared execution substrate through which value
	commitments, impact modeling, safeguards, auditing, and governance interact. Learning
	selects among enabled transitions. Control adjusts parameters and rates. Governance
	legitimately revises structure. At no point is responsibility external to execution
	\cite{Basir2025}.
	
	Because Petri nets model concurrency, accumulation, and feedback explicitly, they provide a
	semantics in which responsibility remains meaningful under parallel adaptation, sustained
	operation, and scale \cite{Murata1989,Reisig2013}. Table 1 summarizes how each layer of the Social Responsibility Stack is realized as a structural component of Petri-net execution semantics. For clarity, sensing is treated as a cross-cutting capability rather than a responsibility layer.
	
	\begin{table}[H]
		\begin{center}
			\begin{tabular}{@{}p{0.28\linewidth}p{0.62\linewidth}@{}}
				\toprule
				\textbf{SRS Layer} & \textbf{Petri-Net Interpretation} \\
				\midrule
				Value Commitments (Value Grounding) &
				Forbidden markings / forbidden places; admissible-marking set $\mathcal{M}_{\mathrm{adm}}$ \\
				Socio-Technical Impact Modeling &
				Critical paths, cycles, siphons/traps; structural vulnerability analysis on the reachability graph \\
				Safeguards and Structural Constraints &
				Guards, inhibitor/permit arcs, capacity places, rate-limit subnets; supervisor constraints on firing \\
				Behavioral Feedback Interfaces &
				Observer/mediator subnets that shape exposure and coupling; interface transitions that modulate loop gain \\
				Monitoring and Auditing &
				Counters, clocks, thresholds, alarms; online estimation of boundary proximity / ``reachability pressure'' \\
				Governance &
				Mode switches / policy places; legitimate structural edits that revise constraints and admissibility \\
				\midrule
				Sensing (cross-cutting, not a layer) &
				State exposure that feeds guards/audits/governance with marking-derived observables \\
				\bottomrule
			\end{tabular}
		\end{center}
		\caption{Petri nets as the execution-level substrate of the Social Responsibility Stack}
		\label{tab:srs-petri-mapping}
	\end{table}
	
	This mapping makes explicit that the SRS is not enforced by optimizing better objectives or by
	adding external oversight after deployment. Instead, it is enforced by shaping the space of
	reachable executions shared across all layers. Value grounding specifies which markings must never
	be reached. Impact modeling reveals how ordinary operation may approach those boundaries.
	Safeguards restructure execution to block such paths. Auditing observes pressure on the
	boundaries. Governance retains authority over the execution structure itself.
	
	Petri nets are therefore not an implementation detail of the SRS. They constitute its execution
	semantics. Through them, responsibility becomes a first-class engineering property: not a
	guideline, not a metric, and not an aspiration, but a structural constraint on what the system is
	allowed to become.
	
	\section{Use Case: Closed-Loop Traffic Management as Responsible Cyber-Physical Execution}
	Urban traffic management systems are canonical cyber--physical systems in which
	algorithmic decision-making is tightly coupled to physical infrastructure,
	human behavior, and institutional constraints. Modern deployments increasingly
	use adaptive control, real-time sensing, and learning-based optimization to
	adjust signal timing, routing recommendations, and priority policies in response
	to evolving traffic conditions.
	
	The operational reality of these systems is not a single control decision,
	but continuous execution under congestion, feedback, and scale. Signal timing
	policies influence driver routing; routing alters congestion patterns; congestion
	reshapes sensed data; and learned controllers adapt based on these endogenous
	signals. Responsibility failures in such systems rarely arise from a single
	unsafe action. Instead, they emerge through cumulative execution trajectories
	that gradually produce globally unacceptable configurations such as gridlock,
	systematic inequity, or emergency vehicle starvation.
	
	Adaptive traffic control systems operate as closed-loop cyber--physical systems.
	Sensor measurements (loop detectors, cameras, vehicle telemetry) feed algorithmic
	controllers that actuate signal phases and routing advisories. These actions
	alter physical traffic flows, which in turn reshape subsequent measurements and
	learning signals.
	
	As in other safety-critical CPS, locally optimal control decisions can accumulate
	into globally unsafe or illegitimate configurations. A traffic system may satisfy
	all local safety constraints while still evolving toward states that undermine
	public safety, equity, or emergency response capacity.
	
	\paragraph{Execution cycle (CPS semantics).}
	We take one control interval as the atomic execution unit. In each interval, the system
	(i) senses queues and arrivals, (ii) applies a timing/routing control action, and
	(iii) induces a physical flow response that changes downstream queues, travel times, and
	subsequent sensing signals. Over repeated intervals, locally stabilizing actions can
	accumulate into regime shifts (persistent congestion patterns, policy lock-in, or
	priority starvation), making responsibility fundamentally a trajectory property rather
	than a single-step constraint.
	
	\paragraph{Places.}
	We define representative places capturing accumulated socio--technical state:
	
	\begin{itemize}
		\item $p_1$: Intersection queue capacity (vehicles accumulated at critical nodes)
		\item $p_2$: Active signal timing policy (current phase allocation in effect)
		\item $p_3$: Institutional reliance on adaptive control (degree of automation trust)
		\item $p_4$: Emergency vehicle priority capacity (ability to preempt signals)
		\item $p_5$: Driver route adaptation (population response to control policies)
		\item $p_6$: Endogenous congestion data (historical traffic shaped by prior control)
	\end{itemize}
	
	\paragraph{Transitions.}
	Representative transitions include:
	
	\begin{itemize}
		\item $t_1$: Sense traffic state and update controller inputs
		\item $t_2$: Apply adaptive signal timing or routing decision
		\item $t_3$: Codify learned policy into default control parameters
		\item $t_4$: Reduce manual oversight due to perceived controller reliability
		\item $t_5$: Driver population adapts routes and departure times
		\item $t_6$: Retrain or retune control policy on endogenous traffic data
	\end{itemize}
	\[
	p_2 \rightarrow t_2 \rightarrow p_3 \rightarrow t_4 \rightarrow p_4
	\rightarrow t_6 \rightarrow p_2
	\]

	\subsection{Responsibility as Forbidden Markings}
	Forbidden markings correspond to physically and institutionally inadmissible CPS
	configurations. A representative template is a marking in which queues and reliance
	have accumulated beyond oversight and emergency-response capacity, e.g.,
	\[
	M(p_1)\ge \bar q,\qquad M(p_3)\ge \bar r,\qquad M(p_4)\le \underline e,
	\]
	(optionally with $M(p_6)\ge \bar d$ indicating strong endogeneity).
	Such configurations capture persistent gridlock, emergency corridor starvation, or
	loss of effective human override. Responsibility requires excluding the corresponding
	set $\mathcal{M}_{\mathrm{bad}}$ from the reachable marking set.

	
	\begin{figure}[t]
		\centering
		\begin{tikzpicture}[
			font=\small,
			>=Stealth,
			place/.style={circle, draw, thick, minimum size=10mm, inner sep=0pt},
			place2/.style={circle, draw, thick, double, minimum size=10mm, inner sep=0pt},
			trans/.style={rectangle, draw, thick, minimum width=3.2mm, minimum height=12mm, fill=gray!10},
			arc/.style={-Stealth, thick},
			lab/.style={font=\scriptsize, align=center},
			dashedbox/.style={draw, dashed, thick, rounded corners=2pt, inner sep=6pt}
			]
			
			\node[place] (p1) at (-7.2,  1.6) {$p_1$};
			\node[lab]   at (-7.2,  0.45) {Queue\\capacity};
			
			\node[place] (p2) at (-3.5,  1.6) {$p_2$};
			\node[lab]   at (-3.5,  0.45) {Active\\policy};
			
			\node[place] (p3) at ( 0.2,  1.6) {$p_3$};
			\node[lab]   at ( 0.2,  0.45) {Institutional\\reliance};
			
			\node[place] (p4) at ( 3.9,  1.6) {$p_4$};
			\node[lab]   at ( 3.9,  0.45) {Emergency\\priority\\capacity};
			
			\node[place] (p5) at ( 3.9, -1.9) {$p_5$};
			\node[lab]   at ( 3.9, -2.9) {Driver\\route\\adaptation};
			
			\node[place] (p6) at ( 0.2, -1.9) {$p_6$};
			\node[lab]   at ( 0.2, -3.5) {Endogenous congestion data};
			
			\node[trans] (t1) at (-5.35, 1.6) {};
			\node[lab]   at (-5.35, 2.65) {$t_1$\\Sense/update};
			
			\node[trans] (t2) at (-1.65, 1.6) {};
			\node[lab]   at (-1.65, 2.65) {$t_2$\\Apply control};
			
			\node[trans] (t4) at ( 2.05, 1.6) {};
			\node[lab]   at ( 2.05, 2.65) {$t_4$\\Reduce oversight};
			
			\node[trans] (t3) at ( 0.2, 3.25) {};
			\node[lab]   at ( 0.2, 4.25) {$t_3$\\Codify defaults};
			
			\node[trans] (t6) at ( 1.95, -1.9) {};
			\node[lab]   at ( 1.95, -0.95) {$t_6$\\Retune/retrain};
			
			\node[trans] (t5) at ( 2.05, -1.9) {};
			\node[lab]   at ( 2.05, -0.5) {$t_5$\\Population adapts};
			
			\draw[arc] (p1) -- (t1);
			\draw[arc] (t1) -- (p2);
			
			\draw[arc] (p2) -- (t2);
			\draw[arc] (t2) -- (p3);
			
			\draw[arc] (p3) -- (t4);
			\draw[arc] (t4) -- (p4);
			
			\draw[arc] (p4) to[bend left=14] (t6);
			\draw[arc] (t6) -- (p2); 
			
			\draw[arc] (p3) to[bend right=18] (t6);
			\draw[arc] (t2) to[bend left=20] (p6); 
			\draw[arc] (p6) -- (t6);               
			
			\draw[arc] (p2) to[bend right=18] (t5);
			\draw[arc] (t5) -- (p5);
			\draw[arc] (p5) to[bend right=18] (t1); 
			
			\draw[arc] (p3) to[bend left=18] (t3);
			\draw[arc] (t3) to[bend left=18] (p2);
			
			\node[dashedbox, fit=(p2)(t2)(p3)(t4)(p4)(t6)(p6), 
			label={[lab]south:Reinforcing loop: reliance $\rightarrow$ oversight erosion $\rightarrow$ \\ endogenous retuning}] (loopBox) {};
			
			\node[lab, anchor=west] at (-7.9, 3.85) {\textbf{CPS semantics:} sensing $\rightarrow$ actuation $\rightarrow$ physical flow};
			\node[lab, anchor=west] at (-7.9, -4.0) {\textbf{Reachability risk:} repeated locally ``good'' firings accumulate into unacceptable markings};
			
		\end{tikzpicture}
		
		\caption{Petri-net execution semantics for adaptive urban traffic management (CPS).
			Places encode accumulated socio--technical conditions (queues, reliance, oversight capacity, endogenous data, and population adaptation);
			transitions encode repeatable events (sensing, control actuation, codification, oversight erosion, adaptation, and retraining).
			Responsibility is enforced by excluding forbidden markings (e.g., emergency starvation, persistent gridlock) from the reachable marking set.}
		\label{fig:traffic_petri}
	\end{figure}
	
	\section{Use Case: Automated Risk Scoring in Public Institutions}
	
	We now develop a concrete use case to show how responsibility, when formalized as
	reachability, becomes operational inside the Social Responsibility Stack (SRS) through
	Petri-net execution semantics. The setting is automated risk scoring in public-sector
	decision-making, including welfare eligibility triage, parole assessment support,
	immigration screening, fraud prioritization, and social-service allocation. These systems
	have been repeatedly analyzed as canonical examples of machine-mediated governance and
	accountability stressors in real institutions \cite{Raji2020,Selbst2019,Rahwan2019}.
	
	The operational reality of these deployments is not a single prediction followed by a
	single decision. It is a repeated workflow embedded in staffing limits, procedural
	deadlines, throughput constraints, and audit requirements. Institutional post-mortems and
	empirical studies show that the most consequential harms rarely arise from isolated
	mispredictions or explicit rule violations. They arise when the system becomes woven into
	the agency’s throughput machinery: queue discipline, override norms, appeal practices, and
	data collection routines shift over months of routine use. Such cumulative drift is a
	central theme in safety engineering and socio-technical analysis \cite{Leveson2011,Perrow1984}.
	
	\subsection{System Description as a Socio-Technical Closed Loop}
	
	Consider a public agency that introduces a learned risk scoring model
	\[
	r(x): \mathcal{X}_{\text{case}} \rightarrow \mathbb{R},
	\]
	where $x$ is a case profile assembled from administrative records (history of benefits,
	prior compliance events, employment records), behavioral signals (missed appointments,
	contact patterns), and historical outcomes used as labels. The agency formally deploys
	$r(x)$ as decision support rather than automated decision-making: the score is presented
	to an officer or caseworker alongside a short explanation and a recommended action band
	(e.g., \emph{low/medium/high}). This “advisory” framing, combined with nominal human
	authority, is standard in public-sector AI deployment narratives \cite{Raji2020,Selbst2019}.
	
	In practice, the score enters an operational pipeline with concrete constraints. Cases
	arrive continuously into a queue. Officers face service-level deadlines and throughput
	targets. Supervisors monitor productivity metrics, consistency across officers, and time-to-resolution.
	A score that appears stable and professionally integrated becomes a coordination device:
	it standardizes triage, compresses deliberation time, and reduces justification burden.
	Over time, officers learn which cases are ``safe'' to defer and which require escalation,
	and managerial policy learns which thresholds improve flow. Similar feedback structures
	have been identified as key drivers of emergent harm in complex socio-technical systems
	\cite{Rahwan2019,Leveson2011}.
	
	The resulting deployment is a closed loop, not a one-shot prediction:
	score $\rightarrow$ officer action $\rightarrow$ procedural routing and documentation
	$\rightarrow$ downstream outcomes $\rightarrow$ new data and labels $\rightarrow$ retraining and
	threshold revision $\rightarrow$ updated scoring behavior. Responsibility failures, when they occur,
	are emergent properties of this loop rather than isolated technical faults, reflecting the
	multi-causal drift mechanisms emphasized in organizational accident theory \cite{Perrow1984}.
	
	\subsection{Why Objective-Based Evaluation Appears Sufficient---and Is Not}
	
	From a standard optimization viewpoint, the deployment can look exemplary. Predictive
	accuracy improves on held-out data, calibration plots remain stable, and fairness
	regularization may reduce measured disparities in error rates or score distributions.
	The agency can point to written policy affirming human oversight and to procedural steps
	for appeals or review. These are precisely the indicators typically used to argue that a
	system behaves correctly under objective-based evaluation \cite{Russell2019}.
	
	However, these indicators fail to capture how the institution itself changes under sustained use. Over time, override rates decline not because officers are coerced, but because deferring to the system is faster and becomes normatively safe once supervisors treat the score as standard practice. Appeals diminish as affected individuals lose access to alternative evidentiary pathways, as documentation practices become score-conditioned, and as time-to-resolution is implicitly organized around score bands. Discretionary review erodes as staff and resources are reallocated from deliberation to throughput.
	
	Simultaneously, data becomes increasingly endogenous: prior system recommendations shape who is investigated, who is sanctioned, and which outcomes are recorded, and those records in turn become the training substrate for subsequent model iterations. Through this process, institutional practices and learning pipelines become mutually reinforcing, producing accountability erosion without any explicit policy change or measurable degradation in local performance. Such dynamics have been repeatedly documented across public-sector AI deployments and are widely recognized as systemic failure modes rather than isolated implementation flaws \cite{Raji2020,Leveson2011}.
	
	At no point must local performance degrade, and at no point must an explicit specification
	be violated. The failure is that the overall execution trajectory makes an unacceptable
	configuration reachable: an institution in which accountability exists on paper but cannot
	be exercised in practice. In formal terms, the failure is not optimization failure but
	admissible-execution failure: a reachability failure \cite{Baier2008}.
	
	\subsection{Petri-Net Representation of the Deployment}
	
	We now formalize the execution semantics of this deployment using a Petri net
	\[
	\mathcal{N} = (P, T, F, M_0),
	\]
	where places represent accumulated socio-technical conditions and transitions represent
	repeatable operational events that advance the pipeline. Petri nets are suited here because
	they make accumulation, concurrency, and feedback explicit rather than forcing them into
	memoryless transitions \cite{Peterson1981,Murata1989}.
	
	\paragraph{Places.}
	We define representative places:
	\begin{itemize}
		\item $p_1$: Human discretionary capacity available (time, training, authority to override)
		\item $p_2$: Algorithmic recommendation invoked (score present in the workflow)
		\item $p_3$: Institutional reliance on automated scores (thresholds, routings, norms)
		\item $p_4$: Oversight and review capacity (appeals staff, audit bandwidth, conference time)
		\item $p_5$: Population behavioral adaptation (strategic compliance, avoidance, changed reporting)
		\item $p_6$: Endogenous data accumulation (labels shaped by prior scoring-mediated actions)
	\end{itemize}
	
	Tokens represent \emph{accumulation} rather than binary flags. A token increase in $p_3$ can
	correspond to procedural codification (templates, memos, routing rules), managerial dashboards
	that implicitly reward alignment with recommendations, or procurement integration that makes the
	score operationally default. A token decrease in $p_1$ can represent time compression, reduced
	training, or cultural discouragement of overrides. These accumulation effects are central to
	socio-technical dynamics and are poorly captured by snapshot constraint checks \cite{Rahwan2019}.
	
	\paragraph{Transitions.}
	Representative transitions include:
	\begin{itemize}
		\item $t_1$: Generate and present risk score (score attached to case file)
		\item $t_2$: Human defers to recommendation (routing/decision consistent with score band)
		\item $t_3$: Procedure updated to codify usage (thresholds, routing rules, documentation templates)
		\item $t_4$: Oversight reduced due to perceived reliability (audit sampling reduced, staff reallocated)
		\item $t_5$: Population adapts behavior to scoring regime (changed reporting, avoidance, gaming)
		\item $t_6$: Model retrained on endogenous data (new version deployed, thresholds adjusted)
	\end{itemize}
	
	Each transition is locally rational, policy-compliant, and ethically defensible in isolation.
	The responsibility problem is not a single unethical step; it is the interaction over time,
	the same structural pattern emphasized in safety-critical system analysis \cite{Leveson2011}.
	Figure~\ref{fig:risk-scoring-petri} depicts the resulting execution semantics and makes the critical reinforcing cycle explicit as a reachability structure rather than a causal narrative.
	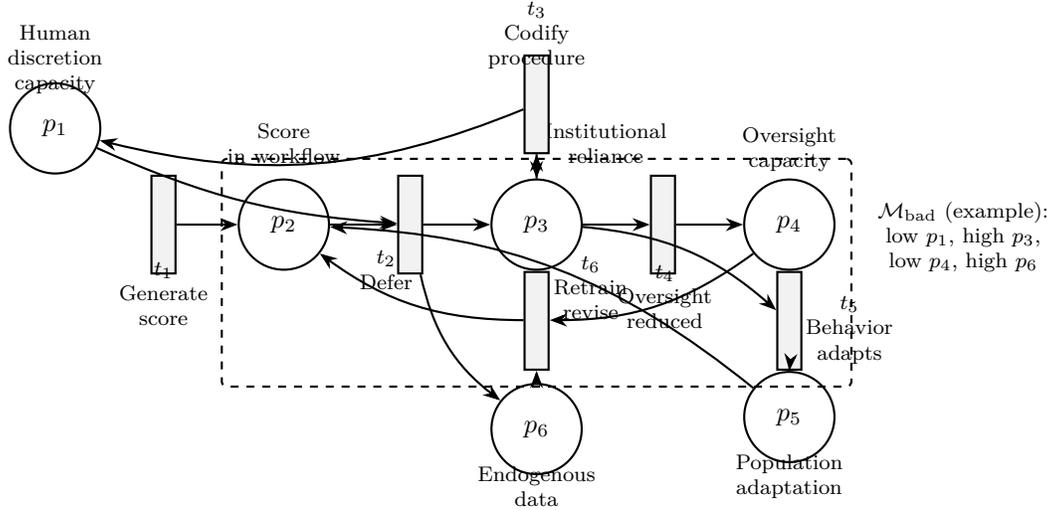
\begin{figure}[t]
		\centering
		\begin{tikzpicture}[scale=0.8,
			>=Stealth,
			font=\small,
			place/.style={circle, draw, thick, minimum size=12mm, inner sep=0pt},
			trans/.style={rectangle, draw, thick, minimum width=3.2mm, minimum height=13mm, fill=gray!10},
			arc/.style={-Stealth, thick},
			lab/.style={font=\scriptsize, align=center},
			dashedbox/.style={draw, dashed, thick, rounded corners=2pt, inner sep=6pt}
			]
			
			\node[place] (p1) at (-8.0,  1.6) {$p_1$};
			\node[place] (p2) at (-4.2,  0.0) {$p_2$};
			\node[place] (p3) at ( 0.0,  0.0) {$p_3$};
			\node[place] (p4) at ( 4.2,  0.0) {$p_4$};
			\node[place] (p6) at ( 0.0, -3.4) {$p_6$};
			\node[place] (p5) at ( 4.2, -3.2) {$p_5$};
			
			\node[lab] at (-8.0,  2.75) {Human\\discretion\\capacity};
			\node[lab] at (-4.2,  1.35) {Score\\in workflow};
			\node[lab] at ( 1.15,  1.35) {Institutional\\reliance};
			\node[lab] at ( 4.2,  1.25) {Oversight\\capacity};
			\node[lab] at ( 0.0, -4.35) {Endogenous\\data};
			\node[lab] at ( 4.2, -4.19) {Population\\adaptation};
			
			\node[trans] (t1) at (-6.2,  0.0) {};
			\node[trans] (t2) at (-2.1,  0.0) {};
			\node[trans] (t4) at ( 2.1,  0.0) {};
			\node[trans] (t6) at ( 0.0, -1.6) {};
			\node[trans] (t3) at ( 0.0,  2.0) {};
			\node[trans] (t5) at ( 4.2, -1.6) {};
			
			\node[lab] at (-6.2, -1.15) {$t_1$\\Generate\\score};
			\node[lab] at (-2.5, -0.75) {$t_2$\\Defer};
			\node[lab] at ( 2.13, -1.2) {$t_4$\\Oversight\\reduced};
			\node[lab] at ( 0.9, -1.05) {$t_6$\\Retrain \\revise};
			\node[lab] at ( 0.0,  3.15) {$t_3$\\Codify\\procedure};
			\node[lab] at ( 5.2, -1.75) {$t_5$\\Behavior\\adapts};
			
			\draw[arc] (t1) -- (p2);
			
			\draw[arc] (p2) -- (t2);
			\draw[arc] (t2) -- (p3);
			\draw[arc] (p3) -- (t4);
			\draw[arc] (t4) -- (p4);
			
			\draw[arc] (p4) to[bend left=18] (t6);
			\draw[arc] (t6) to[bend left=18] (p2);
			
			\draw[arc] (p1) to[bend right=10] (t2);
			
			\draw[arc] (p3) -- (t3);
			\draw[arc] (t3) -- (p3);
			\draw[arc] (t3) to[bend left=18] (p1);
			
			\draw[arc] (t2) to[bend right=20] (p6);
			\draw[arc] (p6) -- (t6);
			
			\draw[arc] (p3) to[bend left=18] (t5);
			\draw[arc] (t5) -- (p5);
			\draw[arc] (p5) to[bend right=18] (p2);
			
			\node[dashedbox, fit=(p2)(t2)(p3)(t4)(p4)(t6)] (loopbox) {};
			\node[lab, anchor=north] at ([yshift=-60pt]loopbox.south)
			{Reinforcing loop: reliance $\rightarrow$ oversight erosion $\rightarrow$ endogenous retraining};
			
			\node[lab, anchor=west] at (5.5, -0.2)
			{$\mathcal{M}_{\text{bad}}$ (example):\\ low $p_1$, high $p_3$,\\ low $p_4$, high $p_6$};
			
		\end{tikzpicture}
		\vspace{-0.2in}
		\caption{Petri-net execution semantics for public-sector risk scoring as a socio-technical closed loop.}
		\label{fig:risk-scoring-petri}
	\end{figure}
	
	\subsection{Feedback Cycles, Drift, and Practical Irreversibility}
	
	The critical structure lies in the feedback cycles. In particular, the loop
	\[
	p_2 \rightarrow t_2 \rightarrow p_3 \rightarrow t_4 \rightarrow p_4 \rightarrow t_6 \rightarrow p_2
	\]
	captures the mechanism by which adaptive control decisions increase institutional reliance,
	erode effective priority and oversight capacity, and feed back into subsequent controller
	retuning.
	
	In the Petri-net execution semantics, the retraining transition $t_6$ is additionally
	\emph{enabled by endogenous congestion data} $p_6$ and may be modulated by observed reliance
	levels $p_3$. This separation distinguishes the core cyber--physical control loop from the
	data-mediated adaptation pathway, while preserving a single reinforcing execution structure.
	
	This loop is not pathological by design; it becomes pathological only under repetition,
	scale, and unchecked accumulation. As reliance ($p_3$) increases, effective emergency
	priority capacity ($p_4$) is reduced: signal preemption becomes less responsive, manual
	intervention is deferred, and retuning proceeds on increasingly endogenous traffic data.
	Over time, the controller adapts to traffic patterns it has itself induced, reinforcing
	congestion regimes that are difficult to reverse without external intervention.
	
	Practical irreversibility is represented by accumulation in places such as $p_6$.
	Once congestion measurements and flow statistics have been shaped for long enough
	by the system’s own control-mediated actions, restoring a pre-deployment baseline
	is no longer a matter of “turning the controller off.” It requires external
	intervention: policy reset, signal reconfiguration, sensor recalibration, or
	institutional override of automated control authority. This mirrors the
	irreversibility of institutional and infrastructural transformations emphasized
	in organizational accident theory \cite{Perrow1984}.

	\subsection{Responsibility as Forbidden Markings}
	
	Let $\mathcal{M}$ denote the set of reachable markings of $\mathcal{N}$. Responsibility is
	specified by declaring forbidden markings
	\[
	\mathcal{M}_{\text{bad}} \subset \mathcal{M},
	\]
	such as configurations in which human discretion is functionally absent, institutional reliance
	exceeds review capacity, or endogenous-data accumulation overwhelms auditability. These are
	non-compensable boundaries on system evolution rather than preferences within it.
	
	Responsibility requires that
	\[
	\mathcal{M}_{\text{bad}} \cap \text{Reach}(\mathcal{N}) = \emptyset,
	\]
	which enforces responsibility structurally, in the same spirit as invariant-based approaches to
	safety and control \cite{Blanchini1999,Ames2017}.
	
	\subsection{Learning Under Structural Constraints}
	
	Adaptation in this deployment corresponds to discovering firing sequences that improve predictive
	utility and operational flow: new thresholds, new routing heuristics, new retraining schedules,
	and new documentation practices. Responsibility-preserving design does not suppress adaptation.
	Instead, it restricts which sequences are admissible so that improvement cannot be purchased by
	trajectory drift into illegitimate configurations.
	
	This distinction mirrors safety-critical practice, where adaptation is permitted only within
	envelopes that preserve invariance. Here, responsibility is enforced not by continual after-the-fact
	correction, but by shaping the space of possible socio-technical futures the institution--system
	coupling may inhabit.
	
	\subsection{Mapping to the Social Responsibility Stack}
	
	This use case maps directly onto the layers of the Social Responsibility Stack. Value grounding
	specifies forbidden markings corresponding to loss of legitimacy or loss of effective agency.
	Socio-technical impact modeling identifies feedback cycles, accumulation pathways, and critical
	reachability paths. Design-time safeguards restrict transitions that enable drift toward forbidden
	regions. Behavioral feedback interfaces observe token accumulation and transition activation
	patterns that signal trajectory drift. Continuous social auditing detects approach toward forbidden
	regions. Governance authorizes structural revision of the net when societal commitments or observed
	risk profiles change.
	
	Responsibility emerges not from any single mechanism, but from controlled reachability across the stack.
	
	\subsection{Interpretation}
	
	This use case shows that responsibility is not a property of isolated decisions, models, or
	objectives. It is a property of execution semantics: what the coupled institution or system is
	allowed to become through repeated operation. Petri nets provide a formalism in which this property
	can be defined, analyzed, and enforced without suppressing learning. The result is not merely safer
	optimization, but governable socio-technical evolution.
	
	\begin{figure}[t]
		\centering
		\begin{tikzpicture}[
			node distance=1.8cm,
			every node/.style={font=\small},
			block/.style={draw, rectangle, rounded corners, align=center, minimum width=3.6cm, minimum height=0.9cm},
			dashedblock/.style={draw, rectangle, rounded corners, dashed, align=center, minimum width=3.6cm, minimum height=0.9cm},
			arrow/.style={->, thick},
			]
			
			\node[block] (society) {Societal References\\
				{\footnotesize norms, laws, institutions}};
			
			\node[block, below of=society] (commitments) {Formalized Commitments\\
				{\footnotesize scoped, accountable}};
			
			\node[block, below of=commitments] (Rset) {Admissible Set $\mathcal{R}$\\
				{\footnotesize derived system constraints}};
			
			\node[block, below of=Rset, yshift=-0.2cm] (system) {Operational Execution\\
				{\footnotesize scoring, routing, retraining}};
			
			\node[dashedblock, right of=system, xshift=4.2cm] (observe) {Social Observability\\
				{\footnotesize behavior, trust, access signals}};
			
			\node[block, below of=observe] (supervise) {Supervisory Control\\
				{\footnotesize constraint tightening, escalation}};
			
			\draw[arrow] (society) -- (commitments);
			\draw[arrow] (commitments) -- (Rset);
			\draw[arrow] (Rset) -- (system);
			
			\draw[arrow] (system.east) -- node[above] {\footnotesize $y(t)$} (observe.west);
			\draw[arrow] (observe.south) -- (supervise.north);
			\draw[arrow] (supervise.west) -- node[below] {\footnotesize $u(t)$} (system.east);
			
			\node[align=left, below of=system, yshift=-0.9cm] (drift) {\footnotesize
				\textbf{Objective:} maintain trajectory invariance\\
				$x(t) \in \mathcal{R}$ under learning and interaction};
			
		\end{tikzpicture}
		\caption{Control-theoretic structure of the Social Responsibility Stack (SRS).
			Societal references are formalized into commitments, translated into an admissible
			set $\mathcal{R}$, and enforced through observability and supervisory control to
			maintain admissible trajectories $x(t)\in\mathcal{R}$ under learning and feedback.}
		\label{fig:srs_control_structure}
	\end{figure}
	
	Figure~\ref{fig:srs_control_structure} summarizes the closed-loop structure: execution produces social signals $y(t)$, and supervisory control applies interventions $u(t)$ to preserve admissibility under sustained operation.
	
	\section{Conclusion}
	
	Responsibility in adaptive socio-technical systems is not a question of intent,
	optimization quality, or post-hoc correction. It is a question of which system
	configurations are reachable under continued operation. Systems fail responsibly
	not because designers intend harm or because objectives are poorly tuned, but
	because execution semantics permit trajectories that gradually exit the region
	of social admissibility.
	
	This paper has argued that treating responsibility as an objective, a preference,
	or an external governance mechanism is structurally insufficient. Optimization
	reasons locally, while oversight reacts episodically. Neither can constrain the
	long-run evolution of systems embedded in feedback-rich environments where harm
	arises through accumulation, interaction, and normalization rather than discrete
	violation.
	
	By formalizing responsibility as a reachability problem, we shift the locus of
	governance from outcome evaluation to trajectory control. The relevant question
	becomes not whether individual decisions are acceptable, but whether the system’s
	execution semantics permit entry into irreversible regimes such as institutional
	lock-in, loss of human agency, erosion of legitimacy, or concentration of power.
	This reframing aligns responsibility with established engineering practice in
	safety-critical domains, where unacceptable states are rendered unreachable by
	design rather than penalized after failure.
	
	Within this framing, Petri nets emerge as a minimal and sufficient execution-level
	formalism for governing responsible behavior. Their explicit treatment of
	concurrency, accumulation, feedback, and irreversibility allows socio-technical
	dynamics to be represented without collapsing them into memoryless transitions or
	scalar objectives. Forbidden markings provide a precise formal expression of
	non-negotiable social commitments, while reachability analysis ensures that
	adaptation remains confined within admissible regions.
	
	Crucially, this approach does not oppose learning, optimization, or innovation.
	It disciplines them. Adaptation remains free to explore within the responsible
	region, discovering improved policies, representations, and strategies. What is
	excluded are destabilizing modes of operation whose harms arise not from isolated
	actions but from cumulative structural effects. Responsibility, in this sense, is
	not a brake on capability but a condition for its sustainable deployment.
	
	The implications for the Social Responsibility Stack are direct. Value grounding
	specifies inadmissible futures. Socio-technical impact modeling identifies the
	paths by which those futures may be approached. Design-time safeguards restrict
	execution semantics. Auditing monitors proximity to responsibility boundaries.
	Governance revises execution structure when societal commitments evolve.
	Responsibility emerges from the coordination of these layers through controlled
	reachability rather than through any single mechanism acting in isolation.
	
	More broadly, this work suggests a shift in how responsible AI should be
	engineered. The central challenge is no longer how to encode values into
	objectives, but how to design systems whose execution semantics preserve
	legitimacy under scale, feedback, and adaptation. Petri nets provide one concrete
	instantiation of this principle; the underlying requirement is more general.
	Responsibility must be internalized as a structural property of system evolution.
	
	In socio-technical systems, what matters is not only what systems do, but what
	they are allowed to become. Engineering responsibility therefore means
	engineering the space of possible futures. This paper provides a formal step in
	that direction.

\end{document}